\documentclass{article}

\title{Computing all $s$-$t$ bridges and articulation points simplified}

\author{Massimo Cairo \thanks{Department of Computer Science, University of Helsinki, Finland, \newline
Email:\{shahbaz.khan,sebastian.schmidt,alexandru.tomescu\}@helsinki.fi}
\and Shahbaz Khan \footnotemark[1]~\thanks{This work was partially funded by the European Research Council (ERC) under the European Union's Horizon 2020 research and innovation programme (grant agreement No.~851093, SAFEBIO).}
\and Romeo Rizzi \thanks{Department of Computer Science, University of Verona, Italy,
Email: romeo.rizzi@univr.it.}
\and Sebastian Schmidt \footnotemark[1]~\footnotemark[2]
\and Alexandru~I.~Tomescu \footnotemark[1]~\footnotemark[2]~\thanks{This work was partially funded by the Academy of Finland (grants No.~322595, 328877).} 
\and Elia Zirondelli \footnotemark[3]~\thanks{
Department of Mathematics, University of Trento, Italy, Email: eliacarlo.zirondelli@unitn.it.}}

\date{}
\usepackage{graphicx}
\graphicspath{{./figures/}}
\usepackage[linesnumbered,ruled,vlined]{algorithm2e}
\SetKw{Continue}{continue}
\usepackage{amsmath}
\usepackage{amssymb}
\usepackage{amsthm}
\usepackage{subcaption}
\usepackage{hyperref}
\usepackage[dvipsnames]{xcolor}
\usepackage[utf8]{inputenc}
\usepackage{float}
\usepackage{cleveref}

\newtheorem{theorem}{Theorem}
\newtheorem{lemma}[theorem]{Lemma}

\newcommand{\tail}{\textsc{tail}}
\newcommand{\head}{\textsc{head}}

\newcommand{\bridge}{$s$-$t$ bridge}
\newcommand{\bridges}{$s$-$t$ bridges}
\newcommand{\articulationpoint}{$s$-$t$ articulation point}
\newcommand{\articulationpoints}{$s$-$t$ articulation points}

\newcommand{\articulationsequence}{articulation sequence}

\begin{document}
	
	\maketitle
	
\begin{abstract}
Given a directed graph $G$ and a pair of nodes $s$ and $t$, an \emph{$s$-$t$ bridge} of $G$ is an edge whose removal breaks all $s$-$t$ paths of $G$. Similarly, an \emph{$s$-$t$ articulation point} of $G$ is a node whose removal breaks all $s$-$t$ paths of $G$. Computing the sequence of all $s$-$t$ bridges of $G$ (as well as the \articulationpoints{}) is a basic graph problem, solvable in linear time using the classical min-cut algorithm~\cite{FordF56}. 
	
	When dealing with cuts of {\em unit} size (\bridges{}) this algorithm can be simplified to {\em a single graph traversal from $s$ to $t$ avoiding an arbitrary $s$-$t$ path, which is {\em interrupted} at the $s$-$t$ bridges}. Further, the corresponding proof is also simplified making it independent of the theory of network flows. \\
		
\noindent
\textbf{Keywords:} reachability, graph algorithm, strong bridge, strong articulation point 

    \end{abstract}
	
	\section{Introduction}
	
	Connectivity and reachability are fundamental graph-theoretical problems studied extensively in the literature \cite{diestel10,HandbookGT,cormen01introduction,skiena}. 
	A key notion underlying such algorithms is that of edges (or nodes) critical for connectivity or reachability. The most basic variant of these are bridges (or articulation points), which are defined as follows. A \emph{bridge} of an undirected graph, also referred as \emph{cut edge}, is an edge whose removal increases the number of connected components. Similarly, a \emph{strong bridge} in a (directed) graph is an edge whose removal increases the number of strongly connected components of the graph. (Strong) articulation points are defined in an analogous manner by replacing edge with node.
	
	Special applications consider the notion of bridges to be parameterised by the nodes that become disconnected upon its removal~\cite{ItalianoLS12,Tarjan76}. Given a node $s$, we say that an edge is an \emph{$s$~bridge} (also referred as \emph{edge dominators} from source $s$~\cite{ItalianoLS12}) if there exists a node $t$ that is no longer reachable from $s$ when the edge is removed. Moreover, given both nodes $s$ and $t$, an \emph{$s$-$t$ bridge} (or $s$-$t$ articulation point) is an edge (or node) whose removal makes $t$ no longer reachable from $s$.

	\subparagraph*{Related work.} 
	For undirected graphs, the classical algorithm by Tarjan~\cite{Tarjan74} computes all bridges and articulation points in linear time. However, for directed graphs only recently Italiano et al.~\cite{ItalianoLS12} presented an algorithm to compute all strong bridges and strong articulation points in linear time. They also showed that classical algorithms~\cite{Tarjan76,GabowT85} compute $s$ bridges in linear time. The $s$ articulation points (or {\em dominators}) are extensively studied resulting in several linear-time algorithms~\cite{AlstrupHLT99,BuchsbaumKRW05,BuchsbaumGKRTW08}. 
	
	The $s$-$t$ bridges were essentially studied as minimum $s$-$t$ cuts in network flow graphs, where an $s$-$t$ bridge is a cut of {\em unit} size. The classical Ford Fulkerson algorithm~\cite{FordF56} can be used to identify the first $s$-$t$ bridge in the residual graph after pushing {\em unit} flow in the network. Moreover, contracting the entire cut to $s$, one can continue finding the next $s$-$t$ bridge and so on. Since $s$-$t$ bridges limit the maximum flow to {\em one}, the algorithm completes in linear time. 
	
	\section{Preliminaries}
	Let $G := (V, E)$ be a fixed directed graph, where $V$ is a set of $n$ nodes and $E$ a set of $m$ edges, with two given nodes $s,t\in V$. 
	Let $G\setminus X$ and $G - X$ denote the result of removing all edges from $X$, and all nodes from $X$ together with their incident edges, respectively.
	Given an edge $e=(u,v)$, $\head{}(e)=v$ denotes its \emph{head} and $\tail{}(e)=u$ denotes its \emph{tail}.

	\newcommand{\BrgS}{B}
    \newcommand{\brgS}{b}
	\newcommand{\ComB}{{\cal C}}
	\newcommand{\comB}{C}
	
	Let $\BrgS=\{\brgS_1,\brgS_2,...,\brgS_{|\BrgS|}\}$ be the set of $s$-$t$ bridges of $G$. By definition, for all $\brgS_i\in \BrgS$ there exists no path from $s$ to $t$ in $G\setminus \brgS_i$ (see \Cref{fig:bridgeCSa}), and all \bridges{} in $\BrgS$ appear on every $s$-$t$ path in $G$. 
	Further, the \bridges{} in $\BrgS$ are visited in the same order by every $s$-$t$ path in $G$.
	
	\begin{lemma}       
	The \bridges{} in $\BrgS$ are visited in the same order by every $s$-$t$ path in $G$.
	\label{lem:bridgeOrder}
	\end{lemma}
    
    \begin{proof}
    It is sufficient to prove that for any $\brgS_i\in  \BrgS$, all $b_j\in \BrgS$ (where $j\neq i$), can be categorised into those which are always visited before $\brgS_i$ and those that are always visited after $\brgS_i$ irrespective of the $s$-$t$ path chosen in $G$. Consider the graph $G\setminus \brgS_i$, observe that every such $\brgS_j$ is either reachable from $s$, or can reach $t$. It cannot fall in both categories as it would result in an $s$-$t$ path in $G\setminus \brgS_i$, which violates $\brgS_i$ being an \bridge{}. Further, it has to be in at least one category by considering any $s$-$t$ path of $G$, where $\brgS_i$ appears either between $s$ and $\brgS_j$ or between $\brgS_j$ and $t$. Hence, those reachable from $s$ in $G\setminus \brgS_i$ are always visited before $\brgS_i$, and those able to reach $t$ in $G\setminus \brgS_i$ are always visited after $\brgS_i$, irrespective of the $s$-$t$ path chosen in $G$.  
	\end{proof}

\begin{figure}[tbh!]
\centering
\includegraphics[width=0.45\textwidth]{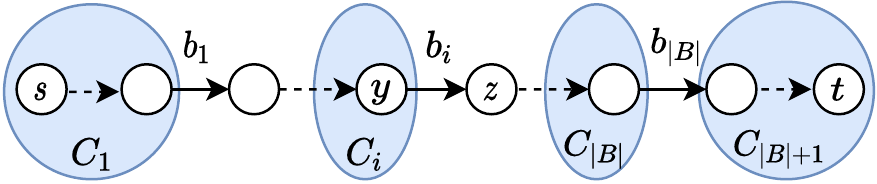}
\caption{Bridge sequence $\BrgS = \{\brgS_1,\dots,\brgS_{|B|}\}$ and corresponding bridge components $\ComB=\{\comB_1,\dots,\comB_{|B|+1}\}$.}
\label{fig:bridgeCSa}
\end{figure}

	Thus, abusing the notation we define $\BrgS$ to be a \emph{ sequence of \bridges{}} ordered by their visit time on any $s$-$t$ path. Also, such a bridge sequence $\BrgS$ implies an increasing part of the graph being reachable from $s$ in $G\setminus\brgS_i$, as $i$ increases. We thus divide the graph reachable from $s$ into \emph{bridge components} $\ComB=\{\comB_1,\comB_2,...,\comB_{|\BrgS|+1}\}$, where $\comB_i$ (for $i\leq |\BrgS|$) denotes the part of graph that is reachable from $s$ in $G\setminus \brgS_i$ but was not reachable in $G\setminus \brgS_{i-1}$ (if any). Additionally, for notational convenience we assume $\comB_{|B|+1}$ to be the part of the graph reachable from $s$ in $G$, but not in $G \setminus \brgS_{|B|}$ (see \Cref{fig:bridgeCSa}). Since bridge components are separated by \bridges{}, every $s$-$t$ path enters $\comB_i$ at a unique vertex ($\head{}(\brgS_{i-1})$ or $s$ for $\comB_1$) referred as its \textit{entry}. Similarly, it leaves $\comB_i$ at a unique vertex ($\tail{}(\brgS_i)$ or $t$ for $\comB_{|B|+1}$) referred as its \textit{exit}.
	
\begin{figure}[tbh!]
\centering
\includegraphics[width=.35\textwidth]{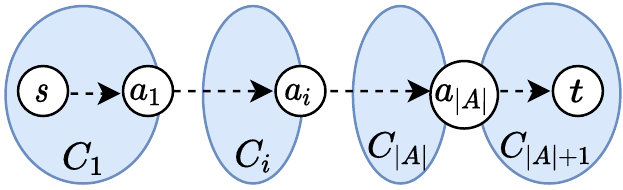}
\caption {Articulation sequence $A =\{a_1,\dots,a_{|A|}\}$ and its components $\ComB=\{\comB_1,\dots,\comB_{|A|+1}\}$ .}
\label{fig:bridgeCSb}
\end{figure}

	Similarly, the $s$-$t$ \emph{articulation points} are defined as the set of nodes $A\subseteq V$, such that removal of any \articulationpoint{} in $A$ disconnects all $s$-$t$ paths in $G$. Thus, $A=\{a_1,a_2,...,a_{|A|}\}$ is a set of nodes such that $\forall a_i\in A$ there exist no path from $s$ to $t$ in $G- a_i$. 
	The \articulationpoints{} in $A$ also follow a fixed order in every $s$-$t$ path (like \bridges{}), so $A$ can be treated as a sequence and it defines the corresponding components~$\ComB$ (see \Cref{fig:bridgeCSb}). Note that the {\em entry} and {\em exit} of an articulation component $\comB_i$ are the preceding and succeeding \articulationpoints{} (if any), else $s$ and $t$ respectively.

    \section{Algorithm}	
    \label{ss:bridgealgo}
    The algorithm can essentially be described as a {\it forward search} from $s$ to $t$ avoiding an arbitrary $s$-$t$ path, which is {\em interrupted} at the $s$-$t$ bridges (or \articulationpoints{}). It discovers the bridge sequence $\BrgS$ (or \articulationsequence~$A$) in order as the search proceeds.
	We first present a linear-time algorithm for \bridges{}, and then extend it to \articulationpoints{}.

	The algorithm first chooses an arbitrary $s$-$t$ path $P$ in $G$. Then it performs a forward search from $s$ to reach $t$ by avoiding the edges of $P$. This search is interrupted by the \bridges{}, since all the \bridges{} lie on $P$ (by definition). If the forward search stops before reaching $t$, it necessarily requires to traverse an edge $\brgS_1$ in $P$ (i.e. the first \bridge{}) to reach $t$. So we continue to forward search from $\head{}(\brgS_1)$ until we stop again to find the next \bridge{}, and so on until we reach $t$.
	When the search is interrupted for the $i^{th}$ time, we look at the last node $y$ on $P$, that was visited by the search. The \bridge{} $\brgS_i$ is then identified as the outgoing edge of $y$ on $P$. 
	
\begin{figure}[tbh!]
\centering
\begin{subfigure}[b]{0.25\textwidth}
\centering
\includegraphics[width=\textwidth]{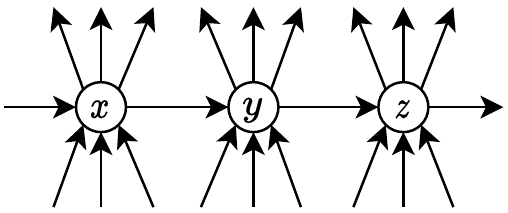}
\end{subfigure}\qquad
\begin{subfigure}[b]{0.25\textwidth}
\centering
\includegraphics[width=\textwidth]{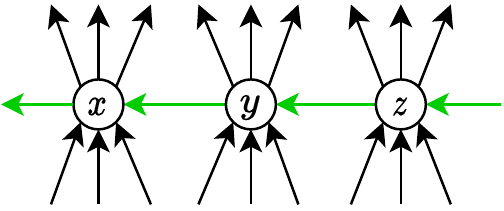}
\end{subfigure}
\caption{Transformation of the graph along $s$-$t$ path for computing $s$-$t$ bridges.}
\label{fig:stTransformB}
\end{figure}

	However, notice that once such a node $y$ is traversed, all the edges on $P$ preceding it are clearly not \bridges{}. Further, the paths starting from these edges may also allow the forward search to proceed further on $P$ beyond $y$. Hence these edges need to be traversed by the forward search before identifying an \bridge{}. This extra procedure can be embedded in the original forward interrupted search of the graph by performing a simple transformation of the graph. Essentially, instead of removing the path $P$ from the graph, we merely reverse it (see \Cref{fig:stTransformB}). Thus, traversing a node on $P$ makes its preceding edges reachable from $s$, ensuring that the traversal is interrupted only on the \bridges{}.

	\begin{algorithm}[tbh]
    	\caption{\textsc{Bridge sequence and Bridge components}}
    	\label{alg:stBC}
    	\KwIn{Graph $G := (V, E)$, $s, t \in V$}
    	\KwOut{Bridge sequence $\BrgS$ and bridge components $\ComB$ and associations $comp[\cdot]$} 
    	\DontPrintSemicolon
    	\BlankLine
    	$P \gets$ Arbitrary $s$-$t$ path in $G$\;
    	$G\gets (G\setminus P) \cup P^{-1}$
    	\tcp*{graph transform, reverse $P$}
    	$i\gets 1$
    	\BlankLine
		\While{$comp[t]=0$}{
		    \lIf{$i=1$}{$Q\gets s$ 
		    \tcp*[f]{initialise search from $s$}}
		    \Else{$y\gets$ Last node on $P$ with $comp[u]\neq 0$\;
		        Add $(y,z)$ to $\BrgS$
		            \tcp*{where $\brgS_i:(y,z)\in P$}
		        $Q\gets z, i\gets i+1$
		    \tcp*{continue search from $z$}
		    }
		    \BlankLine
		    \While(\tcp*[f]{forward search}){$Q\neq \emptyset$}{
		    $u\gets$ Remove node from $Q$\;
		    \ForAll{$(u,v)\in E$ where $comp[v]=0$}
		        {   Add $v$ to $Q$ and $\comB_i$\;
					$comp[v]\gets i$ \;
				}
			}
		}
	\end{algorithm}

	We now formally describe the algorithm (refer to the pseudocode in \Cref{alg:stBC}). After choosing an arbitrary $s$-$t$ path $P$, it transforms the graph as described above, by reversing the path $P$. Along with computing the bridge sequence $\BrgS$ and bridge components $\ComB$, we also ensure access to the component of a node $v$.  
	This is stored in $comp[v]$ which is initialised to $0$, also serving as an indicator that $v$ is not visited. Thereafter, it initiates the search with a queue $Q$ containing $s$, the entrance of $\comB_1$. The forward search continues removing and visiting nodes in $Q$, and adding their unvisited out-neighbours back to $Q$, until $Q$ becomes empty and the search stops. Every node $v$ visited during this search is assigned to $\comB_i$ and has $comp[v]=i$.  If $t$ is not visited yet, the last node $y$ in $P$ that was visited by the search is identified (i.e. the \emph{exit} of $\comB_i$). The \bridge{} $\brgS_i$ is identified to be the outgoing edge of $y$ in $P$ (say $(y,z)$, see \Cref{fig:bridgeCSa}), and added to $\BrgS$. The forward search is then continued from $z$ (i.e. the \emph{entrance} of $\comB_{i+1}$) by adding it to $Q$, and so on. Otherwise, if $t$ was already visited when the search stopped, we terminate the algorithm with bridge sequence $\BrgS$ and the bridge components and node associations in $\ComB$ and $comp[\cdot]$ respectively.

	\section{Analysis and Correctness}
	The algorithm essentially performs five steps which need to be analysed. 
	\textit{Firstly,} 
	computing an $s$-$t$ path $P$ which requires $O(m+n)$ time, by using standard search procedures like DFS or BFS traversal. 
	\textit{Secondly,} transforming the graph which essentially adds and removes $|P|$ edges each, requiring $O(n)$ time. 
	\textit{Thirdly,}  performing the interrupted forward search which only visits the previously unvisited nodes. Hence it requires total $O(m+n)$ time like a simple BFS traversal using the Queue $Q$ and the inner loop.
	\textit{Fourthly,} identifying the last visited node on $P$ which traverses $P$ only once in the whole algorithm, requiring total $O(n)$ time. 
	\textit{Finally,} updating $\BrgS$, $\ComB$ and $comp[t]$ requires to visit the outer loop for each \bridge{}, requiring total $O(n)$ time as the number of \bridges{} and nodes are $O(n)$. Thus, the algorithm requires overall $O(m+n)$ time to compute all the \bridges{} and the associated components.

	The correctness of the algorithm can be proven by maintaining the following invariant.
		 
    \noindent
    \textbf{Invariant $\cal I:$} \emph{	In the transformed graph, the forward search started from the {\em entrance} of $\comB_i$ 
	\begin{enumerate}
	    \item[$(a)$] visits exactly the nodes in $\comB_i$, and 
	    \item[$(b)$] stops along the path $P$ exactly on the \emph{exit} of $\comB_i$. 
	\end{enumerate}
    }
	\begin{proof}
	The graph transformation only affects the paths passing through the edges of $P$, as the remaining edges are unaffected. Hence, to prove ${\cal I}(a)$ it is sufficient to prove that the nodes on $P$ within $\comB_i$, say $v_1$ (entrance)$,v_2,...,v_k$(exit), are reachable from the entrance $v_1$. We prove it by induction over the nodes $v_j$, where the base case ($j=1$) is trivially true. 
	
	\begin{figure}[tbh!]
	\centering
	\includegraphics[width=.45\textwidth]{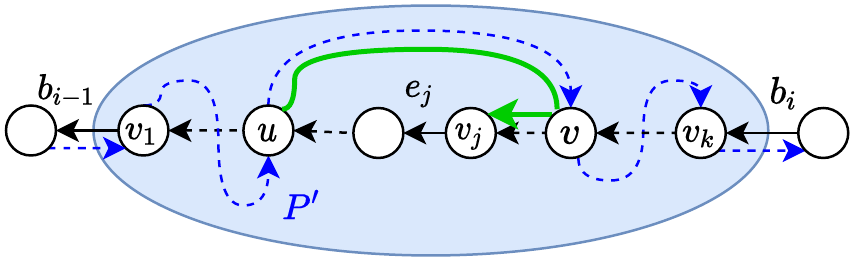}	\caption{Reachability of $v_j$ from $v_1$ (entrance) in the transformed graph. The path $P'$ (blue) is the alternate $s$-$t$ path avoiding $e_j$, which leaves $P$ before $e_j$ for the last time at $u$, then joins $P$ again after $e_j$ at $v$. The subpath of $P'$ from $u$ to $v$ followed by reverse path of $P$ to $v_j$ (green), reaches $v_j$ from $u$ and hence $v_1$.}
		\label{fig:stInvariant}
	\end{figure}

	For any $v_j$ (see \Cref{fig:stInvariant}), assuming nodes up to $v_{j-1}$ are reachable from the entrance, we consider the edge $e_j=(v_{j-1},v_{j})$. Since $e_j$ is not an \bridge{}, there exists a path $P'$ from $s$ to $t$ (and hence from entrance $v_1$ to exit $v_k$) without using $e_j$ in the original graph. 
	Let $u$ and $v$ respectively be the nodes at which $P'$ leaves $P$ for the last time before $e_j$, and the node at which $P'$ joins $P$ again after $e_j$.
	The nodes $u$ and $v$ necessarily exist as the path $P'$ passes through $v_1$ and $v_k$. Note that the subpath of $P'$ from $u$ to $v$ does not pass through any edge in $P$ by definition. Now, $u$ ($=v_j',j'<j$) is reachable from $v_1$ in the transformed graph by induction hypothesis, and there is a path from $v$ to $v_j$ in the reverse path of $P$. Thus, $v_j$ is reachable from $u$ (and hence $v_1$) through $P'$ in the transformed graph. 
	
	Using induction, we have every $v_j$
	and hence the entire $\comB_i$ is reachable from $v_1$, proving ${\cal I}(a)$.
    Further, since $\brgS_i$ is an \bridge{} which is reversed and hence removed in the transformed graph, the forward search cannot reach $\head{}(\brgS_i)$. This is because there is no other edge from $\comB_i$ to $\comB_{i+1}$ including the reversed edges of $P$. Hence, the forward search stops exactly at the exit $v_k$ along $P$ proving ${\cal I}(b)$.
	Now, by induction assuming $\comB_{i-1}$ is computed correctly, all $\comB_i$ would also be computed correctly proving ${\cal I}$ for all $\comB_i$. 
    \end{proof}

  \section{Extension for $s$-$t$ articulation points.}

\begin{figure}[tbh!]
\centering
\begin{subfigure}[b]{0.25\textwidth}
\centering
\includegraphics[width=\textwidth]{stTransformO}
\vspace{.5em}
\end{subfigure}\qquad
\begin{subfigure}[b]{0.3\textwidth}
\centering
\includegraphics[width=\textwidth]{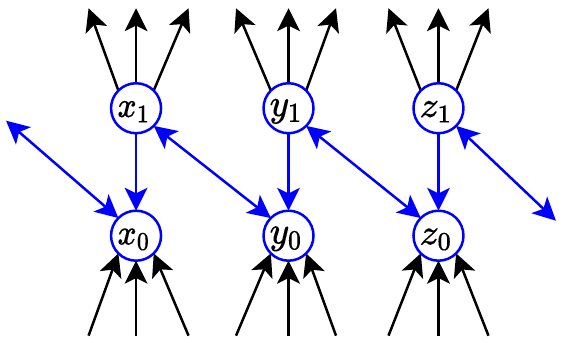}
\end{subfigure}
\caption{Transformation of the graph along $s$-$t$ path for computing $s$-$t$ articulation points.}
\label{fig:stTransformA}
\end{figure}
	
	In order to compute them efficiently, we can use the same algorithm with a different graph transformation. Each node $x$ on $P$ is split into two nodes $x_0$ and $x_1$, having all incoming edges now incoming to $x_0$ and all outgoing edges now outgoing from $x_1$. Further, we have an edge from $x_0$ to $x_1$, which is the \textit{internal edge} of the node $x$. This transformation maintains the path $P$ where each node is split into two by an internal edge. Thus, in the new graph the internal edges of the \articulationpoints{} also act as \bridges{}. Now, when the previous transformation is applied, it reverses the $s$-$t$ path $P$ thereby reversing the internal edges $(x_0,x_1)$ as well. Further, to prevent the search to interrupt at the original \bridges{}, the non-internal edges of $P$ are added back to $G$ (see \Cref{fig:stTransformA}). 

	On executing the same algorithm on the new graph, it reports all the \bridges{} of the modified graph, i.e., the internal edges of the \articulationpoints{}. Also, the components $\ComB$ reported are the corresponding components of $A$. The new transformation adds $|P|=O(n)$ nodes and edges to the graph, and the correctness and analysis follow the same arguments.
	Thus, we have the following theorem. 
	
	\begin{theorem}
	\label{t:linearst}
	Given a graph $G := (V, E)$ with $n$ nodes, $m$ edges and $s,t\in V$, there exists an algorithm to compute all $s$-$t$ bridges (or articulation points) along with their component associations, in $O(m+n)$ time.
	\end{theorem}
	
	\section{Conclusions}
	\label{s:conclusions}
	We have presented a simple algorithm for computing $s$-$t$ bridges  (or articulation points) along with their component associations. Further, it has a simpler proof and is easier to use in practice compared to \cite{FordF56}. 
	
	\bibliography{main}
	
	\end{document}